\documentclass{llncs}

\usepackage{amssymb,amsmath}
\usepackage{upgreek}
\usepackage[all]{xy}
\usepackage{code}

\newcommand{\mylongtitle}{Stack-Summarizing Control-Flow Analysis \\ of Higher-Order Programs}

\newcommand{\myshorttitle}{}

\newcommand{\defterm}[1]{\textbf{#1}}

\newcommand{\ie}{\emph{i.e.}}
\newcommand{\eg}{\emph{e.g.}}

\newcommand{\syn}[1]{\mathsf{#1}}
\newcommand{\var}[1]{\mathit{#1}}
\newcommand{\s}[1]{\mathit{#1}}

\newcommand{\parto}{\rightharpoonup}
\newcommand{\dom}{\var{dom}}

\newcommand{\set}[1]{\left\{#1\right\}}
\newcommand{\setbuild}[2]{\left\{ #1 : #2\right\}}

\newcommand{\Pow}[1]{{\mathcal{P}\left(#1\right)}}
\newcommand{\PowSm}[1]{{\mathcal{P}(#1)}}

\newcommand{\union}{\cup}

\newcommand{\vect}[1]{\langle #1\rangle}
\newcommand{\vecp}[1]{\vec{#1}\;'}

\newcommand{\To}{\mathrel{\Rightarrow}}

\newcommand{\wt}{\sqsubseteq}

\newcommand{\join}{\sqcup}

\newcommand{\StackAlpha}{\Gamma}

\newcommand{\stackchar}{\gamma}

\newcommand{\sembr}[1]{\ensuremath{[\![{#1}]\!]}}

\newcommand{\opor}{\mathrel{|}}

\newcommand{\produces}{\mathrel{::=}}

\newcommand{\vv}{v}

\newcommand{\lam}{\ensuremath{\var{lam}}}

\newcommand{\lamterm}{$\lambda$-term}
\newcommand{\lc}{$\lambda$-calculus}

\newcommand{\call}{\ensuremath{\var{call}}}

\newcommand{\cif}[3]{\mathbf{if}\; #1\; \mathbf{then}\; #2\; \mathbf{else}\; #3}

\newcommand{\free}{\mathit{free}}

\newcommand{\ttlp}{\mbox{\tt (}}
\newcommand{\ttrp}{\mbox{\tt )}}

\newcommand{\appform}[2]{\ttlp #1\; #2\ttrp}
\newcommand{\lamform}[2]{\ttlp \uplambda\;\ttlp#1\ttrp\;#2\ttrp}

\newcommand{\letiform}[3]{\ttlp {\tt let}\; \ttlp\ttlp#1\; #2\ttrp\ttrp\; #3\ttrp}

\newcommand{\fexpr}{f}
\newcommand{\expr}{e}
\newcommand{\aexpr}{\mbox{\sl {\ae}}}

\newcommand{\Eval}{{\mathcal{E}}}
\newcommand{\ArgEval}{{\mathcal{A}}}

\newcommand{\Inject}{{\mathcal{I}}}

\newcommand{\state}{\varsigma}

\newcommand{\qstate}{q}

\newcommand{\Addr}{\s{Addr}}

\newcommand{\store}{\sigma}

\newcommand{\env}{\rho}

\newcommand{\clo}{\var{clo}}
\newcommand{\cont}{\kappa}

\newcommand{\alloc}{\mathit{alloc}}

\newcommand{\addr}{a}

\newcommand{\aTo}{\leadsto}

\newcommand{\aInject}{{\hat{\mathcal{I}}}}

\newcommand{\sa}[1]{\widehat{\mathit{#1}}}

\newcommand{\aEval}{{\hat{\mathcal{E}}}}
\newcommand{\aArgEval}{{\hat{\mathcal{A}}}}

\newcommand{\astate}{{\hat{\varsigma}}}

\newcommand{\aAddr}{\sa{Addr}} 

\newcommand{\astore}{{\hat{\sigma}}}
\newcommand{\aenv}{{\hat{\rho}}}

\newcommand{\aclo}{{\widehat{\var{clo}}}}

\newcommand{\acont}{{\hat{\kappa}}}

\newcommand{\aaddr}{{\hat{\addr}}}
\newcommand{\aalloc}{{\widehat{alloc}}}

\newcommand{\aval}{{\widehat{\var{val}}}}

\newcommand{\absmap}{\alpha}

\newcommand{\rp}{rp}
\newcommand{\arp}{\widehat{\rp}}
\newcommand{\nulladdr}{\sa{null}}

\newcommand{\ControlStates}{Q}
\newcommand{\transfunction}{\delta}

\newcommand{\conf}{c}
\newcommand{\aconf}{{\hat c}}

\newcommand{\phrame}{\phi}
\newcommand{\aphrame}{\hat{\phi}}

\newcommand{\stackact}{g}

\newcommand\PDTrans{\longmapsto}

\DeclareMathOperator*{\pdedge}{\rightarrowtail}

\newcommand{\biedge}[2]{#1 \mathrel{\rightarrowtail} #2}

\newcommand{\DSConfs}{S}
\newcommand{\DSEdges}{E}
\newcommand{\ECEdges}{H}

\newcommand{\fnet}[1]{\lfloor #1 \rfloor}

\newcommand{\ecg}{$\epsilon$-closure graph}

\newcommand\RPDTrans{{\;\longmapsto\!\!\!\!\!\!\!\!\longrightarrow}}

\newcommand{\rootset}{\mathit{Root}}
\newcommand{\touch}{\mathcal{T}}
\newcommand{\touchrel}[1]{\leadsto_{\touch,#1}}
\newcommand{\reach}{\mathcal{R}}
\newcommand{\agc}{AGC}

\newcommand{\aTopd}{\rightharpoondown}

\newcommand{\aTons}{\curvearrowright}%\smile\!\!\!>}
\DeclareMathOperator*{\aoTons}{\aTons}

\newcommand{\aToss}{\approx\!\!>}
\DeclareMathOperator*{\aoToss}{\aToss}

\newcommand{\aToagc}{\aToss_{AGC}}

\newcommand{\dsg}{G}

\newcommand{\epcg}{G_\epsilon}

\newcommand{\delconf}{\Delta S}
\newcommand{\deledge}{\Delta E}
\newcommand{\delshort}{\Delta H}

\newcommand{\addEmpty}{addShort}
\newcommand{\addEdge}{addEdge}
\newcommand{\Explore}{Explore}

\newcommand{\ass}{\widehat{ss}}
\newcommand{\ssum}{\absmap_{S}}
\newcommand{\wts}{\wt_{S}}
\newcommand{\apush}{\textbf{push}}

\newcommand{\ssumfs}{\ssum^{fs}}
\newcommand{\wtsfs}{\wts^{fs}}
\newcommand{\apushfs}{\apush^{fs}}

\newcommand{\ssumra}{\ssum^{ra}}
\newcommand{\wtsra}{\wts^{ra}}
\newcommand{\apushra}{\apush^{ra}}

\newcommand{\apred}{\mathbf{pred}}

\begin{document}
\title{\mylongtitle}
\titlerunning{\myshorttitle}
\author{Christopher Earl\inst{1} \and Matthew Might\inst{1} \and David {Van Horn}\inst{2}}
%
%\authorrunning{Matthew Might}   % abbreviated author list (for running head)
%
\institute{University of Utah,
Salt Lake City, Utah
\and 
Northeastern University,
Boston, Massachusetts}

\maketitle              % typeset the title of the contribution

\begin{abstract}
  Two sinks drain precision from higher-order flow analyses: (1) merging
of argument values upon procedure call and (2) merging of return
values upon procedure return.
To combat the loss of precision, these two sinks have been addressed
independently.  
In the case of procedure calls, \emph{abstract garbage collection}
reduces argument merging; while in the case of procedure returns,
\emph{context-free approaches} eliminate return value merging.
It is natural to expect a combined analysis could enjoy the mutually
beneficial interaction between the two approaches.
The central contribution of this work is a direct product of abstract
garbage collection with context-free analysis.
%
%% The resulting analysis alleviates argument-merging even as it
%% eliminates return-flow merging.
%
The central challenge to overcome is the conflict between the core
constraint of a pushdown system and the needs of garbage collection:
a pushdown system can only see the top of the stack, yet garbage
collection needs to see the entire stack during a collection.
To make the direct product computable, we develop ``stack summaries,''
a  method for tracking stack properties at each control
state in a pushdown analysis of higher-order programs.

% Local Variables: 
% TeX-master: "paper-lncs"
% End: 

% LocalWords:  emph paper-sigplan

\end{abstract}

\section{Introduction} \label{intro}

In higher-order flow analysis~\cite{mattmight:Shivers:1991:CFA},
merging is the enemy.
Merging of flow sets and control-flow paths is what destroys
precision.
Merging occurs in two forms: merging on call (for arguments) and
merging on return (for return-flow and return values).
Our goal is to alleviate argument-merging while simultaneously
eliminating return-flow merging.

For an example of both kinds of merging, consider the following code:
\begin{code}
(let* ((id (lambda (x) x))
       (a  (id 3))
       (b  (id 4)))
 b)\end{code}
Flow-sensitive 0CFA makes the following inferences:
%% \begin{enumerate}
%% \item {\tt (lambda (x) x)} flows to {\tt id}; 
%% \item {\tt (id 3)} invokes {\tt (lambda (x) x)};
%% \item {\tt 3} flows to {\tt x};
%% \item {\tt (lambda (x) x)} returns to {\tt (id 3)}; 
%% \item {\tt 3} flows to {\tt a};
%% \item {\tt (id 4)} invokes {\tt (lambda (x) x)};
%% \item {\tt 4} flows to {\tt x}; 
%% \item {\tt (lambda (x) x)} returns to {\tt (id 4)};
%% \item {\tt 3} flows to {\tt b}; 
%% \item {\tt 4} flows to {\tt b}; 
%% \item {\tt 3} or {\tt 4} returns from the program.
%% \end{enumerate}
%
%In this example, 
(1) the two instances of the argument {\tt x} merge together:
{\tt 3} and {\tt 4}, and
(2) the (implicit) continuations at applications of the identity
function also merge together, causing its return values to merge in
the variable {\tt b}.

For two decades, context-sensitivity---splitting bindings, calls and
returns among a finite set of abstract instances---has been the
``solution'' to both merging problems.
But, context-sensitivity is a finite, monotonic band-aid for an
infinite, non-monotonic problem.
Arguments ultimately merge because flow information accretes
\emph{monotonically}: once an analysis says that $x$ may flow to $y$, it will
never revoke that inference.
Return-flows merge because finite flow analyses implicitly allocate a
finite number of abstract stack pointers to continuations.

\subsection{Two solutions}

Might and Shivers developed abstract garbage collection (abstract GC)
to tame the argument-merging
problem~\cite{mattmight:Might:2006:GammaCFA}.
Abstract GC assumes a small-step abstract
interpretation~\cite{mattmight:Cousot:1977:AI,mattmight:Cousot:1979:Galois}
over a finite state-space.
Much like concrete GC, abstract GC finds all of the reachable
addresses in an abstract heap and reclaims any unreachable
addresses.
With abstract GC, the abstract heap no longer grows
monotonically across a small-step transition: the same abstract address has
the chance to get rebound to a singleton flow set over a different
value many times over, thereby making more judicious use of the
abstract resources available.
For programs composed of (possibly recursive) tail calls and closures
which never escape, abstract garbage collection delivers perfectly
precise control-flow analysis.

Pushdown control-flow analysis
(PDCFA)~\cite{local:Earl:2010:PDCFA}, a relative of Vardoulakis
and Shivers's CFA2~\cite{mattmight:Vardoulakis:2010:CFA2}, solves the
return-flow problem by using the arbitrarily large pushdown stack to
model the concrete call stack; thus, continuations never merge.
PDCFA can reason through arbitrary levels of recursive calls.
%% , be they
%% direct, indirect or tail.

\subsection{One problem}
Our mission is to combine the benefits of both abstract
garbage collection and pushdown control-flow analysis: to produce an
``almost complete'' control-flow analysis which eliminates \emph{most}
argument merging and \emph{all} continuation merging. 
\emph{
The challenge is an apparent incompatibility between the two techniques.}

Abstract garbage collection must have the ability to search an
entire state---stack included---to determine the reachable addresses.
A pushdown control-flow analysis approximates the evaluation of a
program, roughly speaking, as a push-down automaton.  The machine
states of the PDA represent the control string, environment, and store
(heap) of the evaluator; while the stack of the PDA represents the
evaluator's stack, where each letter of the stack alphabet represents
a continuation frame.  Transitions of the PDA push and pop frames much
like an abstract machine (e.g., the CESK machine) pushes and pops
continuations.  

When a machine like the CESK machine performs garbage collection, it
crawls the stack to determine reachable heap locations.
That works because the stack is explicit in each machine state:
it's the K component.
But in a pushdown analysis, the abstract stack is not represented in
each control state.
Rather, the stack's structure is scattered across the transition
graph between control-states.
In more detail, the data structure accumulated during pushdown
analysis is a transition graph where each node contains the C, E and S
components and each edge is labeled with the change to K that
happens on that transition.
In order to recover the possible stack(s) at a node in the graph, the
analysis must consider all the paths from the initial control
state to the current state.

% Yet, pushdown control-flow analysis spreads the structure of the stack
% across the abstract transition graph between control-states.
% %
% To know the possible stacks for a state, one must know all the paths
% from the initial control state to the current state.

\subsection{Our contribution: SSCFA}

To complete our mission, we develop a new kind of higher-order
pushdown-like control-flow analysis that includes stack
\emph{summaries} in its control states: SSCFA.
To make our contribution more general, we place constraints on stack
summaries (in lieu of fixing them to be reachable addresses) and we
let clients supply alternate summaries, \eg, all procedures live on the
stack, whether the security context is privileged or unprivileged.
Thus, SSCFA could drive pushdown variants of dependence analysis or
even escape analysis in addition to abstract garbage collection.

%\subsection{Overview}
%

The remainder of this paper is organized as follows:
\begin{itemize}
\item Section~\ref{prelim} reviews simple preliminaries for working
  with pushdown systems.
\item Section~\ref{pdcfa} reviews pushdown control-flow analysis and
  Dyck state graphs.
\item Section~\ref{abstract-gc} introduces the problem with integrating abstract garbage collection and pushdown analysis.
\item Section~\ref{ssintro} informally introduces the notion of a
  \emph{stack summary}, defines criteria for \emph{stack
    summarization}, and gives example summarization strategies.
\item Section~\ref{sscfa} formally defines stack-summarizing control
  flow-analysis.
\item Section~\ref{ssgammacfa} presents the computable product of stack
  summarizing control-flow analysis and abstract garbage collection.
\item Section~\ref{related} discusses related work and Section~\ref{concl} concludes.
\end{itemize}

% The term reachability has been overloaded with multiple meanings.
% %
% This paper and the analysis presented here deals with two types of reachability.
% %
% Hence we use the phrase, ``reachable reachability'' to draw attention to this fact.
% %
% First, the base analysis is pushdown control-flow analysis as formulated~\cite{PDCFA}.
% %
% This analysis finds only the \emph{reachable} states of a given program, thus avoiding an exhaustive state-space search.
% %
% Second, we optimize this analysis by limiting the heap to only \emph{reachable} addresses and values.
% %
% Previously, in \cite{GammaCFA} this heap optimization was shown to improve precision and runtime (in practice) for classical control-flow analysis.

% For the heap reachability optimization to be sound, all reachable addresses must be known.
% %
% For classical control-flow analysis, the reachable addresses are readily available.
% %
% Unfortunately, pushdown control-flow analysis distributes the reachable addresses throughout the transition system.
% %
% This paper presents an extension of pushdown control-flow analysis that collects the reachable addresses in an efficient manner, so that the heap reachability optimization can be used with pushdown control-flow analysis.

% Local Variables: 
% TeX-master: "paper-lncs"
% End: 

\section{Pushdown preliminaries} \label{prelim}

In this work, we make extensive use of pushdown systems.
(A pushdown automaton is a specific kind of pushdown system.)
There are many (equivalent) definitions of these machines in the
literature, so we adapt our own definitions from \cite{mattmight:Sipser:2005:Theory}.
Even those familiar with pushdown theory may want to skim this
section to pick up our notation.

\subsection{Stack actions, stack change and stack manipulation}

Stacks are sequences over a alphabet $\StackAlpha$.
Pushdown systems do much stack manipulation; to represent this more
concisely, we turn stack alphabets into ``action'' sets where each
character represents a stack change: push, pop or no change.

For each character $\stackchar$ in a stack alphabet $\StackAlpha$, the
\defterm{stack-action} set $\StackAlpha_\pm$ contains a push
($\stackchar_{+}$) and a pop ($\stackchar_{-}$) character and a
no-stack-change indicator ($\epsilon$):
\begin{align*}
  \stackact \in \StackAlpha_\pm &\produces \epsilon && \text{[stack unchanged]} 
  \\
  &\;\;\opor\;\; \stackchar_{+}  \;\;\;\text{ for each } \stackchar \in \StackAlpha && \text{[pushed $\stackchar$]}
  \\
  &\;\;\opor\;\; \stackchar_{-}  \;\;\;\text{ for each } \stackchar \in \StackAlpha && \text{[popped $\stackchar$]}
  \text.
\end{align*}
% In this paper, the symbol $\stackact$ represents some stack action.

Given a string of stack actions, we can compact it into a minimal
string describing net stack change.
We do so through the operator $\fnet{\cdot} : \StackAlpha_\pm^* \to
\StackAlpha_\pm^*$, which cancels out opposing adjacent push-pop stack
actions:
\(
  \fnet{\vec{\stackact} \; \stackchar_+\stackchar_- \; \vecp{\stackact}} = 
  \fnet{\vec{\stackact} \; \vecp{\stackact}} 
\) 
and
\(
  \fnet{\vec{\stackact} \; \epsilon \; \vecp{\stackact}} = 
  \fnet{\vec{\stackact} \; \vecp{\stackact}} 
\),
so that
  \(\fnet{\vec{\stackact}} = \vec{\stackact}\),
if there are no cancellations to be made in the string $\vec{\stackchar}$.
%
% We can convert a net string back into a stack by stripping off the
% push symbols with the stackify operator, $\fstackify{\cdot} :
% \StackAlpha^{*}_\pm \parto \StackAlpha^*$:
% \begin{align*}
%   \fstackify{\stackchar_+ \stackchar_+' \ldots \stackchar_+^{(n)}} =
%   \vect{\stackchar^{(n)}, \ldots, \stackchar', \stackchar}
%  \text.
% \end{align*}
% Notice the stackify operator is defined for strings containing
% only push actions. 

%% The function $\fTop : \StackAlpha^* \parto \StackAlpha$ returns the
%% top of a stack (if any):
%% \begin{equation*}
%%   \fTop\vect{\stackchar_1,\ldots,\stackchar_n} = \stackchar_1
%%   \text.
%% \end{equation*}

% (Might's 
% work
% on $\Delta$CFA~\cite{mattmight:Might:2006:DeltaCFA,mattmight:Might:2007:DeltaCFA,mattmight:Might:2007:Dissertation},
% studies the theory of these ``frame strings'' in detail.)

\subsection{Pushdown systems}
A \defterm{pushdown system} is a triple
$M = (\ControlStates,\StackAlpha,\transfunction)$ where
$\ControlStates$ is a finite set of control states;
$\StackAlpha$ is a stack alphabet; and
$\transfunction \subseteq
  \ControlStates \times \StackAlpha_\pm \times \ControlStates$ is a transition relation.
We use $\mathbb{PDS}$ to denote the class of all pushdown systems.
Unlike the more widely known pushdown automaton, a pushdown system
\emph{does not recognize a language}.
\\

\noindent
For the following definitions, let $M = (\ControlStates,\StackAlpha,\transfunction)$.
The \defterm{configurations} of this
  machine are pairs over control states and
  stacks:
  \(\s{Configs}(M) = \ControlStates \times \StackAlpha^*\).
The labeled \defterm{transition relation} $(\PDTrans_{M}) \subseteq \s{Configs}(M) \times \StackAlpha_\pm \times 
  \s{Configs}(M)$ determines whether one configuration may transition to another while performing the given stack action:
  \begin{align*}
    (\qstate, \vec{\stackchar}) 
    \mathrel{\PDTrans_M^\epsilon}
    (\qstate',\vec{\stackchar}) 
    & \text{ iff }
    %(\qstate,\epsilon,\qstate')
    (\qstate,\epsilon,\qstate')
    \in \transfunction
    && \text{[no change]}
\\
   (\qstate, \stackchar' : \vec{\stackchar}) 
   \mathrel{\PDTrans_M^{\stackchar'_{-}}}
   (\qstate',\vec{\stackchar})
    & \text{ iff }
    %(\qstate,\stackchar'_{-},\qstate') 
    (\qstate,\stackchar'_{-},\qstate')
    \in \transfunction
    && \text{[pop]}
\\
    (\qstate, \vec{\stackchar}) 
    \mathrel{\PDTrans_{M}^{\stackchar'_{+}}}
    (\qstate',\stackchar' : \vec{\stackchar}) 
    & \text{ iff }
    %(\qstate,\stackchar'_{+},\qstate')
    (\qstate,\stackchar'_{+},\qstate')
    \in \transfunction
    && \text{[push]}
    \text.
  \end{align*}
Additionally, we define:
\begin{align*}
  \conf \mathrel{\PDTrans_{M}} \conf' &\text{ iff }
  \conf \mathrel{\PDTrans_{M}^{\stackact}} \conf' 
  \text{ for some stack action } \stackact\text,
  \\
  \conf \mathrel{\PDTrans_M^{\vec{\stackact}}} \conf'
  &\text{ iff }
  \conf = \conf_0
  \mathrel{\PDTrans_M^{\stackact_1}} 
  \conf_1
  % \mathrel{\PDTrans_M^{\stackact_2}} 
  \cdots
  % \mathrel{\PDTrans_M^{\stackact_{n-1}}} 
  \conf_{n-1} 
  \mathrel{\PDTrans_M^{\stackact_n}}
  \conf_n = \conf'
  \text{ for some  $\vec{\stackact} = \stackact_1 \ldots
    \stackact_n$,}% and configurations $\conf_0,\ldots,\conf_n$}
  \\
  \conf \mathrel{\PDTrans_M^{*}} \conf'
    &\text{ iff }
    \conf \mathrel{\PDTrans_M^{\vec{\stackact}}} \conf'
    \text{ for some }
    \vec{\stackact}
    \text.
\end{align*}

% A \defterm{legal configuration-path} through system $M$ is a sequence $\vec{\conf} =
% \vect{\conf_1,\ldots,\conf_n}$ where adjacent configurations are a
% legal transition:
% \begin{equation*}
%   \conf_i \mathrel{ \PDTrans_M} \conf_{i + 1}
%   \text.
% \end{equation*}
%

%% \paragraph{Note}
%% Some texts define the transition relation $\transfunction$ so that 
%% $\transfunction \subseteq \ControlStates \times \StackAlpha \times \ControlStates
%% \times \StackAlpha^*$.
%% %
%% In these texts, $(\qstate,\stackchar',\qstate',\vec{\stackchar}) \in
%% \transfunction$ means, ``if in control state $\qstate$ while the
%% character $\stackchar'$ is on top, pop the stack, transition to
%% control state $\qstate'$ and push $\vec{\stackchar}$.''
%% %
%% Clearly, we can convert between these two representations by
%% introducing extra control states to our representation when it needs
%% to push multiple characters.

\subsection{Rooted pushdown systems}

A \defterm{rooted pushdown system} is a quadruple 
$(\ControlStates,\StackAlpha,\transfunction,\qstate_0)$ in which 
$(\ControlStates,\StackAlpha,\transfunction)$ is a pushdown system and
$\qstate_0 \in \ControlStates$ is an initial (root) state.
$\mathbb{RPDS}$  is the class of all rooted pushdown
systems.
For a rooted pushdown system $M =
(\ControlStates,\StackAlpha,\transfunction,\qstate_0)$, we define a
the \defterm{root-reachable transition relation}:
\begin{equation*}
  \conf \RPDTrans_M^{\stackact} \conf' \text{ iff } 
  (\qstate_0,\vect{})
  \mathrel{\PDTrans_M^*} 
  \conf 
  \text{ and }
  \conf 
  \mathrel{\PDTrans_M^\stackact}
  \conf'
  \text.
\end{equation*}
In other words, the root-reachable transition relation also makes
sure that the root control state can actually reach the transition.
The root-reachable relation is overloaded to operate on
control states:
\begin{equation*}
  \qstate 
   \mathrel{\RPDTrans_M^\stackact}
  \qstate' \text{ iff }
  (\qstate,\vec{\stackchar}) 
   \mathrel{\RPDTrans_M^\stackact}
  (\qstate',\vecp{\stackchar}) 
  \text{ for some stacks }
  \vec{\stackchar},
  \vecp{\stackchar}
  \text.
\end{equation*}

\section{Pushdown control-flow analysis} \label{pdcfa}

In this section we present the concrete and abstract semantics for the
pushdown control-flow analysis (PDCFA) of a call-by-value \lc{}, which
represents the core of a higher-order programming language.
To simplify presentation of the concrete and abstract semantics, we
analyze programs in A-Normal Form (ANF),
%% %
%% However, the choice of A-Normal Form is strictly cosmetic;
%% all of our results can be replayed \emph{mutatis mutandis} in a
%% direct-style setting.
%
%% For deeper discussion of small-step semantics and abstract
%% interpretation, we refer the reader to Van Horn and
%% Might~\cite{dvanhorn:VanHorn2010Abstracting}.
% 
a syntactic discipline that enforces an order of evaluation and
requires that all arguments to a function be atomic:
\begin{align*}
  \expr \in \syn{Exp} &\produces \letiform{\vv}{\call}{\expr} && \text{[non-tail call]}
  \\
  &\;\;\opor\;\; \call && \text{[tail call]}
  \\
  &\;\;\opor\;\; \aexpr && \text{[return]}
  \\
  \fexpr,\aexpr \in \syn{Atom} &\produces \vv \opor \lam && \text{[atomic expressions]}
  \\
  \lam \in \syn{Lam} &\produces \lamform{\vv}{\expr} && \text{[lambda terms]}
  \\
  \call \in \syn{Call} &\produces \appform{\fexpr}{\aexpr} && \text{[applications]}
  \\
  \vv \in \syn{Var} &\text{ is a set of identifiers} && \text{[variables]}
  \text.
\end{align*}
  
We use the CESK machine \cite{mattmight:Felleisen:1987:CESK} to
specify the semantics of ANF.
We chose the CESK machine because it has an explicit stack.
Figure~\ref{fig:conc-abs-conf-space} contains the concrete
configuration-space of this machine.
Each configuration contains a control-state component consisting of an
expression, an environment and a store; and a continuation/stack
component.
Under our abstractions, the stack component of this
configuration-space becomes both a finite ``stack summary'' in
abstract control states and a stack component in the pushdown system.
(See Appendix~\ref{classical-cfa} for a review of the
\emph{finite-state} approach and comparison to the pushdown approach.)

PDCFA does not collapse the abstract
stack into a finite structure like classical control-flow analysis.
%OLD: Pushdown control-flow analysis (PDCFA) takes a completely different approach to managing the stack than classical control-flow analysis.
%
Instead of folding the stack into the store through frame pointers,
PDCFA distributes the stack throughout an enriched abstract transition
system.
The abstract configuration-space of pushdown control-flow analysis (Figure
\ref{fig:conc-abs-conf-space}) is similar to concrete formulation.

% Foremost, we are not working with the configuration-space directly;
% rather we deal with the control-state-space.
% %
% Hence, a configuration is now defined as a state and a stack paired
% together.
% %
% The results of pushdown control-flow analysis are rooted pushdown
% systems rather than nondeterminisitic finite automata.
% %
% A rooted pushdown system handles the stack and configurations
% implicitly, so we use the control-state-space instead of the
% configuration-space.

% The next change is that the store only contains bindings from
% addresses to closures, instead of from addresses to closures and
% frames.
% %
% The abstraction of the store creates imprecision.
% %
% By keeping the frames out of the store (and in a precise structure) we
% avoid this imprecision for continuations.

% The final difference is that frames no longer contain return pointers.
% %
% Again, this is because the enriched abstract transition system
% encapsulates that information precisely.

% To conserve space and the emphasize the similarity between the
% concrete and abstract semantics, we present the two in parallel.
\begin{figure}
  \begin{align*}
    \conf \in \s{Conf} &= \s{State} \times \s{Kont} 
    &
    \aconf \in \sa{Conf} &= \sa{State} \times \sa{Kont}
    && \text{[configurations]}
    \\
    \state \in \s{State} &= \syn{Exp} \times \s{Env} \times \s{Store}
    &
    \astate \in \sa{State} &= \syn{Exp} \times \sa{Env} \times \sa{Store}
    && \text{[states]}
    \\
    \env \in \s{Env} &= \syn{Var} \parto \s{Addr} 
    &
    \aenv \in \sa{Env} &= \syn{Var} \parto \sa{Addr}
    && \text{[environments]}
    \\
    \store \in \s{Store} &= \s{Addr} \to \s{Clo} 
    &
    \astore \in \sa{Store} &= \sa{Addr} \to \Pow{\sa{Clo}} 
    && \text{[stores]}
    \\
    \clo \in \s{Clo} &= \syn{Lam} \times \s{Env} 
    &
    \aclo \in \sa{Clo} &= \syn{Lam} \times \sa{Env}
    && \text{[closures]}
    \\
    \cont \in \s{Kont} &= \s{Frame}^*
    &
    \acont \in \sa{Kont} &= \sa{Frame}^* 
    && \text{[stacks]}
    \\
    \phrame \in \s{Frame} &= \syn{Var} \times \syn{Exp} \times \s{Env}  
    &
    \aphrame \in \sa{Frame} &= \syn{Var} \times \syn{Exp} \times \sa{Env}
    && \text{[stack frames]}
    \\
    \addr \in \s{Addr} &\text{ is an infinite set}
    &
    \aaddr \in \sa{Addr} &\text{ is a \emph{finite} set}
    && \text{[addresses]}
  \end{align*}
  \caption{Configuration-space for CESK machine and pushdown control-flow analysis.}
  \label{fig:conc-abs-conf-space}
\end{figure}

%% \begin{figure}
%% \begin{align*}
%%   \aconf \in \sa{Conf} &= \sa{State} \times \sa{Kont} && \text{[configurations]}
%%   \\
%%   \astate \in \sa{State} &= \syn{Exp} \times \sa{Env} \times \sa{Store} && \text{[states]}
%%   \\
%%   \aenv \in \sa{Env} &= \syn{Var} \parto \sa{Addr} && \text{[environments]}
%%   \\
%%   \astore \in \sa{Store} &= \sa{Addr} \to \Pow{\sa{Clo}} && \text{[stores]}
%%   % \\
%%   % \aden \in \sa{Den} &= \Pow{\sa{Clo}} && \text{[denotable values]}
%%   \\
%%   \aclo \in \sa{Clo} &= \syn{Lam} \times \sa{Env} && \text{[closures]}
%%   \\
%%   \acont \in \sa{Kont} &= \sa{Frame}^* && \text{[stacks]}
%%   \\
%%   \aphrame \in \sa{Frame} &= \syn{Var} \times \syn{Exp} \times \sa{Env} && \text{[stack frames]}
%%   \\
%%   \aaddr \in \sa{Addr} &\text{ is a \emph{finite} set of addresses} && \text{[addresses]}
%%   \text.
%% \end{align*}
%% \caption{Abstract configuration-space for pushdown control-flow analysis.}
%% \label{fig:abs-pd-state-space}
%% \end{figure}

\subsection{Concrete semantics and PDCFA}

Next, we define the concrete semantics of ANF and pushdown
control-flow analysis simultaneously.
Specifically, we define program-to-machine injection, atomic
expression evaluation, reachable configurations/control states, the
transition relation and a resource-allocation parameter.
The abstraction functions that connect the concrete
configuration-space to the abstract configuration-space are
straightforward structural abstraction functions.
(Formal definitions of these abstractions can be found in Appendix
~\ref{abstraction-functions}.)

\paragraph{Program injection}

The concrete program-injection function pairs an expression with an
empty environment, store and stack to create the
initial configuration:
\begin{equation*}
  \conf_0 = \Inject(\expr) = (\expr, [], [], \vect{})
  \text.
\end{equation*}

We define two abstract injection functions---one that produces an
initial abstract control state, and one that produces an initial
abstract configuration.
The control-state injector $\aInject_\astate : \syn{Exp}
\to \sa{State}$ pairs an expression with an empty environment and
store to create the initial abstract state:
\begin{equation*}
  \astate_0 = \aInject_\astate(\expr) = (\expr, [], [])
  \text.
\end{equation*}
The configuration injector $\aInject_\aconf : \syn{Exp} \to \sa{Conf}$ tacks on an empty stack:
\begin{equation*}
  \aconf_0 = \aInject_\aconf(\expr) = (\aInject_\astate(\expr), \vect{})
  \text.
\end{equation*}

\paragraph{Atomic expression evaluation}

The atomic expression evaluator, $\ArgEval : \syn{Atom} \times \s{Env}
\times \s{Store} \parto \s{Clo}$ (or, $\aArgEval : \syn{Atom} \times
\sa{Env} \times \sa{Store} \to \PowSm{\sa{Clo}}$ in the abstract),
returns the value of an atomic expression in the context of an
environment and a store:
\begin{align*}
  \ArgEval(\lam,\env,\store) &= (\lam,\env)
  &
  \aArgEval(\lam,\aenv,\astore) &= \set{(\lam,\env)} 
  && \text{[closure creation]}
  \\
  \ArgEval(\vv,\env,\store) &= \store(\env(\vv))
  &
  \aArgEval(\vv,\aenv,\astore) &= \astore(\aenv(\vv)) 
  && \text{[variable look-up]}
  \text.
\end{align*}

\paragraph{Reachable configurations}
The program evaluator $\Eval : \syn{Exp} \to
\Pow{\s{Conf}}$ (or, $\aEval : \syn{Exp} \to \PowSm{\sa{Conf}}$ in the
abstract) computes all of the configurations reachable from the
initial configuration:
\begin{align*}
  \Eval(\expr) = \setbuild{ \conf }{ \Inject(\expr) \To^* \conf } 
  &&
  \aEval(\expr) = \setbuild{ \aconf }{ \aInject_\aconf(\expr) \aTopd^* \aconf }
  \text.
\end{align*}
Since the stack's depth is unbounded, the number of reachable
configurations in both the concrete \emph{and} abstract semantics
could be infinite.

\paragraph{Transition relation}
The concrete transition, $\conf \To \conf'$, and its abstract
counterpart, $\aconf \aTopd \aconf'$, each have three rules.
The first rule handles tail calls by evaluating the function into a
closure, evaluating the argument into a value and then moving to the
body of the \lamterm{} within the closure:
\begin{gather*}
  \conf = (\overbrace{(\sembr{\appform{\fexpr}{\aexpr}}, \env, \store)}^{\state}, \cont)\
  \To\
  ((\expr,\env'',\store'),\cont)
  \text{, where }
  \\
  \begin{align*}
    (\sembr{\lamform{\vv}{\expr}}, \env') &= \ArgEval(\fexpr,\env,\store)
    &
    \env'' &= \env'[\vv \mapsto \addr]
    \\
    \addr &= \alloc(\vv,\state)
    &
    \store' &= \store[\addr \mapsto \ArgEval(\aexpr,\env,\store)]
  \end{align*}
  \\[.5em]
  \aconf = (\overbrace{(\sembr{\appform{\fexpr}{\aexpr}}, \aenv, \astore)}^{\astate}, \acont)\
  \aTopd\
  ((\expr,\aenv'',\astore'), \acont)
  \text{, where }
  \\
  \begin{align*}
    (\sembr{\lamform{\vv}{\expr}}, \aenv') &\in \aArgEval(\fexpr,\aenv,\astore)
    &
    \aenv'' &= \aenv'[\vv \mapsto \aaddr]
    \\
    \aaddr &= \aalloc(\vv,\astate)
    &
    \astore' &= \astore \join [\aaddr \mapsto \aArgEval(\aexpr,\aenv,\astore)]
    \text.
  \end{align*}
\end{gather*}
In the abstract semantics, the tail-call transition is
nondeterministic, since multiple abstract closures may be invoked.

A non-tail call builds a frame, adds it to the stack, and evaluates the call:
\begin{align*}
  ((\sembr{\letiform{\vv}{\call}{\expr}}, \env, \store), \cont)\
  &\To\
  ((\call,\env,\store), (\vv,\expr,\env) : \cont)
  \\[.5em]
  ((\sembr{\letiform{\vv}{\call}{\expr}}, \aenv, \astore), \acont)\
  & \aTopd\
  ((\call,\aenv,\astore), (\vv,\expr,\aenv) : \acont)
  \text.
\end{align*}
A function return pops the top frame of the stack and uses that frame
to continue the computation after binding the return value to the
frame's variable:
\begin{gather*}
  \conf = (\overbrace{(\aexpr, \env, \store)}^{\state}, (\vv,\expr,\env') : \cont)\
  \To\
  ((\expr,\env'',\store'), \cont)
  \text{, where }
  \\
  \begin{align*}
    \addr &= \alloc(\vv,\state)
    &
    \env'' &= \env'[\vv \mapsto \addr] 
    &
    \store' &= \store[\addr \mapsto \ArgEval(\aexpr,\env,\store)]
 \end{align*}
  \\[1em]
  \aconf = (\overbrace{(\aexpr, \aenv, \astore)}^{\astate}, 
    (\vv,\expr,\aenv')  : \acont'  
  )\
  \aTopd\
  ((\expr,\aenv'',\astore'), \acont')
  \text{, where }
  \\
  \begin{align*}
    \aaddr &= \aalloc(\vv,\astate)
    &
    \aenv'' &= \aenv'[\vv \mapsto \aaddr]
    &
    \astore' &= \astore \join [\aaddr \mapsto \aArgEval(\aexpr,\aenv,\astore)]
   \text.  
  \end{align*}
\end{gather*}

\paragraph{Allocation, polyvariance and context-sensitivity}
The address-allocation function is an opaque parameter in both
semantics.
For the concrete semantics, letting addresses be
natural numbers suffices, and then the allocator can use the lowest
unused address: $\s{Addr} = \mathbb{N}$ and
$\alloc(v,(\expr,\env,\store,\cont)) = 1 + \max(\dom(\store))$.
%% \begin{align*}
%%   \s{Addr} &= \mathbb{N}
%%   &
%%   \alloc(v,(\expr,\env,\store,\cont)) &= 
%%   1 + \max(\dom(\store))
%%   \text.
%% \end{align*}
%
The opacity is useful because abstract semantics also parameterize
allocation---to provide a knob to tune the polyvariance and
context-sensitivity of the resulting analysis---and allowing the
abstract semantics to choose a particular concrete allocation function
can simplify proofs of soundness.
%
%% Like in the abstract semantics for classical control-flow analysis,
%% the abstract allocation function $\aalloc : \syn{Var} \times \sa{Conf}
%% \to \sa{Addr}$ determines the polyvariance and context-sensitivity of
%% the analysis and is still a parameter to the analysis.

\subsection{Removing the explicit stack}

The reachable subset of the abstract configuration-space for any
program could be infinite.
(There is no bound on the depth of the stack, so there are an infinite
number of stacks and therefore an infinte number of configurations.)
Consequently, the na\"ive exploration of the reachable abstract
configurations used in classical flow analyses may not terminate.
Fortunately, because the abstract semantics describe a pushdown
system, we can construct a finite (computable) description of the
reachable configurations.
Specifically, we can construct a labeled transition system in which
nodes are control states, and labels on edges denote stack change.

We define the legal transitions in any such graph through the transition
relation $(\aTons) \subseteq \sa{State} \times \sa{Frame}_\pm \times
\sa{State}$.
Three rules define this relation: one determining when to push, one
when to pop and the last when to leave the stack unchanged.
The labels on each transition are the stack action for the transition:
%
%A tail call makes no changes to the stack and so is labeled with an $\epsilon$:
\begin{align*}
  \astate
  \ \aoTons^\epsilon\ 
  \astate'
  & \text{ iff }
  \aconf = (\astate, \acont)
  \ \aTopd\ 
  (\astate', \acont) = \aconf'
  \mbox{, for any stack }\acont
&& \text{[tail call]}
\\
%
% A non-tail call pushes a frame, $\aphrame$, onto the stack and so is labeled with a $\aphrame_+$:
  \astate
  \ \aoTons^{\aphrame_+}\
  \astate'
  & \text{ iff }
  \aconf = (\astate, \acont)
  \ \aTopd\ 
  (\astate', \aphrame : \acont) = \aconf'
  \mbox{, for any stack }\acont
&& \text{[non-tail call]}
\\
%
% A function return pops a frame, $\aphrame$, from the stack and so is labeled with a $\aphrame_-$:
  \astate
  \ \aoTons^{\aphrame_-}\ 
  \astate'
  & \text{ iff }
  \aconf = (\astate, \aphrame : \acont)
  \ \aTopd\ 
  (\astate', \acont) = \aconf'
  \mbox{, for any stack }\acont
&& \text{[return]}
\text.
\end{align*}

From this transition relation, we build a rooted pushdown system
$(\ControlStates,\StackAlpha,\transfunction,\qstate_0)$ for a program $\expr$ such that $\ControlStates = \sa{States}$,
$\StackAlpha = \sa{Frame}$,
$\transfunction = (\aTons)$, and
$\qstate_0 = \aInject_\astate(\expr)$.
The subset of this rooted pushdown system reachable over legal paths
provides a finite description of the original configuration-space.
This finite subset is a \textbf{Dyck state graph}
(DSG)~\cite{local:Earl:2010:PDCFA}.
(A path is legal only if all of the pops match up with pushes; there can be unmatched pushes left over.)
Several techniques can compute the Dyck state graph; for an efficient
technique specific to PDCFA, we defer to our recent
work~\cite{local:Earl:2010:PDCFA} or the algorithm as modified in
Appendix~\ref{sec:algorithm}.

% Local Variables: 
% TeX-master: "paper-lncs"
% End: 

\section{Adding abstract garbage collection to PDCFA}  \label{abstract-gc}

In the classical version of abstract garbage collection, the abstract
interpretation ``collects'' each configuration before each transition~\cite{mattmight:Might:2006:GammaCFA}.
To collect a state, it explores the state to find the
reachable abstract addresses, and then it discards unreachable
addresses from the store, \ie, it maps them to the empty set.

Suppose we were to add abstract garbage collection to PDCFA.
At first, we might try collecting a control state prior to adding an
edge.
But, this approach doesn't work: to know the reachable addresses of a
configuration, the analysis must have access to the stack paired with
the control state.
Unfortunately, the stack has been distributed across the
Dyck state graph being accreted during the analysis.
To determine the possible stacks paired with a control state, the
analysis must consider all legal paths to that control state.

Considering all possible paths to a control state is expensive, and
troublesome in any event, since there could be an infinite number of
such paths.
A better solution would allow the analysis to iteratively compute
properties of stacks---like reachable addresses---and store these
summaries at individual control states.
We call this solution stack summaries.

\section{Stack summaries} \label{stack-summaries} \label{ssintro}

As PDCFA constructs a DSG, it accretes reachable control states one
edge at a time.
Each time it adds a labeled edge, it is abstractly executing the
transition relation $(\aTopd)$.
To perform abstract garbage collection before each transition, the
analysis must know the reachable addresses for all configurations
described by paths to that state.
To accomplish this, we add a stack summary to each control state.
A stack summary is a client-defined finite abstraction of a stack.
To perform abstract garbagae collection, we will instantiate this
summary to be the reachable addresses in the stack.

A stack summary describes some property of the stack, \eg, the topmost
frame, the reachable addresses, the privilege level of the current
context.
With respect to our analysis, the set $\sa{Summary}$ is a parameter
containing all stack summaries, and we denote an individual stack
summary as $\ass$.
A summarizing function, $\ssum: \s{Stack} \to \sa{Summary}$,
walks a stack to compute a summary.
Every stack summary regime also requires a push function parameter,
$\apush : \sa{Frame} \times \sa{Summary} \to \sa{Summary}$, which
computes the abstract effect of pushing a frame on a summary.
There are three requirements on stack summaries:
\begin{enumerate}
\item Summaries must be able to represent all possible stacks.
\item The set of summaries must be finite.
\item Summaries must form a lattice ($\wts$).%, such that $\ass \wts \ass'$ means that the stacks represented by $\ass$ are a subset of those represented by $\ass'$.
% \item Each summary must retain useful information about the stacks it represents.
% \item There must be a way to transition between summaries that mimics the transitions between concrete stacks.
\end{enumerate}
In addition, the push function must faithfully simulate
concrete push; formally:
\begin{equation*}
  \text{ if }
  \absmap(\phrame) \wt \aphrame
  \text{ and }
%  \absmap(\cont) \wt \acont
  \ssum(\cont) \wts \ass
\text{, then }
 \ssum(\phrame : \cont) \wt \apush(\aphrame,\ass)
  \text.
\end{equation*}

% All but the last of there requirements depend upon the application:
% %
% What an acceptable number of stack summaries is depends upon how much precision the application needs.
% %
% What information is useful depends entirely upon the application.
% %
% Thus maintaining consistent transitions is the most general requirement because there are only three ways to transition between stacks.

% To transition from one stack to another, we have three operations:
% %
% We push a frame onto the top, pop the top frame, or do nothing.
% %
% (It may be said that the no-op operation is not an operation by definition, but since there are three types of transtions between configurations/states it is useful to have three stack operations that match these transitions nicely.)
% %
% Both making no change and popping a frame can be handled in a very general manner.

We can efficiently percolate stack summaries through the construction
of a Dyck state graph, so that the algorithm never has to reconsider all paths
to a control state.
In fact, the algorithm never considers an entire path all at once; it
propagates summaries edge-by-edge.
To extend the Dyck-state-graph-construction algorithm, we need to
consider three cases: what is the effect of adding a push; what is the
effect of adding a pop; and what is the effect of a stack no-op?
In this section, we describe the core of the algorithm informally
but with sufficient detail to motivate the high-level idea.
In the next section, we'll describe the system-space of the algorithm
formally, and Appendix~\ref{sec:algorithm} contains the algorithm for
computing a Dyck state graph with stack summaries.

\subsection{Propagating stack summaries during DSG construction}

The propagation of stack summaries across no-op and pop edges during
DSG construction is agnostic of the particular stack summary in use.
Propagation across push edges is fully factored into the push function
parameter, $\apush$.

\paragraph{The summarizing no-op operation}

When the Dyck state graph construction algorithm needs to propagate
summaries across an edge which does not change the stack, the new
stack summary is identical to the old sumary:
when there is no stack change, there is no change to stack summaries.
%
% So, for a tail-call transition or a shortcut edge, neither of which change the stack,
% %
% %% \begin{equation*}
% %%   \xymatrix{
% %%     \state \ar[r]^{\epsilon}
% %%     &
% %%     \state'
% %%   }
% %% \end{equation*}
% %
% every stack $\cont$ possible at the first state $\state$ is also possible at the second state $\state'$ and therefore their summaries are identical:
% %
% \begin{align*}
%   (\state', \cont) \in \aEval(\expr)
%   &\text{ iff }
%   (\state, \cont) \in \aEval(\expr)
%   \\
%   \ssum(\cont)
%   &=
%   \ssum(\cont)
% \end{align*}
% %
% So both the stacks and their summaries are the same in both configurations.

\paragraph{The summarizing pop operation}

The pop operation, like the no-op operation, can be handled without
knowledge of the particular stack summary in use.
In PDCFA, every pop transition has at least one matching push
transition.
The stack summaries after a pop are those stack summaries that can
reach the new state with no net stack change.

These states are easy to find, because the DSG construction algorithm
maintains an $\epsilon$-closure graph in addition to the control-state
transition graph.
Edges in the $\epsilon$-closure graph connect states reachable
through no net stack change.

Diagramatically, we know that the stack summary at state
$\state_4$ in the following is the same as the stack summary for
$\state_1$:
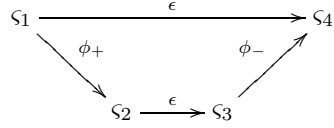
\begin{equation*}
  \xymatrix{
    \state_1 \ar[dr]^{\phrame_+}
    \ar[rrr]^{\epsilon}
    &
    &
    &
    \state_4
    \\
    &
    \state_2 \ar[r]^{\epsilon}
    &
    \state_3 \ar[ur]^{\phrame_-}
  }
\end{equation*}

\paragraph{The summarizing push operation}

% The push operation of stacks does not generalize like the no-op operation and the pop operation have.
% %
Pushing a frame onto a stack makes a local change to the stack.
% %
% (The bottom frame is still the bottom; the second from bottom is still the second from bottom; etc., etc.)
% %
However, pushing a frame onto a stack may nontrivially change the
summary.
The operation $\apush$ must be able to determine the new stack
summary, so that when a push edge is introduced, the $\apush$ operation
determines the subsequent summary.

\subsection{Example: A frame-set summary}

The frame-set summary is both general and useful.
The frame-set summary is the set of (abstract) frames currently in the
stack:
\begin{equation*}
  \sa{Summary_{fs}} = \Pow{\sa{Frame}}
  \text.
\end{equation*}
This summary ignores order and repetition in favor of finite size and a simple (subset-based) lattice:
\begin{equation*}
  \ass \wtsfs \ass'
  \text{ iff }
  \ass \subseteq \ass'
  \text.
\end{equation*}

The summarization function for the frame set summary,
$\ssumfs : \sa{Stack} \to \sa{Summary_{fs}}$,
abstracts each frame and keeps it in a set:
%
% \begin{equation*}
%   \ssumfs(\cont) = \left\{
%   \begin{array}{ll}
%     \set{\absmap(\phrame)}
%     \cup
%     \ssumfs(\cont')
%     &
%     \text{if } \cont = \phrame : \cont'
%     \\
%     \emptyset
%     &
%     \text{if } \cont \text{ is an empty stack} 
%   \end{array}
%   \right.
% \end{equation*}
\begin{equation*}
  \ssumfs\vect{\phrame_1,\ldots,\phrame_n} =
  \set{\absmap_{Frame}(\phrame_1),\ldots,\absmap_{Frame}(\phrame_n)}
  \text.
\end{equation*}
%
%Every stack can be recursed down to the empty stack, so every stack has a representative summary.

The push operation $\apushfs : \sa{Frame} \times \sa{Summary_{fs}} \to \sa{Summary_{fs}}$ simply adds the new frame to the set:
\(
  \apushfs(\aphrame, \ass)
  =
  \{\aphrame\}
  \cup
  \ass
  \).
%
% The constraint on the push operation is true again by the various definitions:
% %
% \begin{align*}
%   \set{\absmap(\phrame)}
%   \cup
%   \ssumfs(\cont)
%   &=
%   \ass
%   =
%   \ssumfs(\phrame : \cont)
%   &
%   \ass
%   &\wtsfs
%   \ass
%   \\
%   \set{\absmap(\phrame)}
%   \cup
%   \ssumfs(\cont)
%   &=
%   \ass
%   =
%   \apushfs(\absmap(\phrame), \ssumfs(\cont))
%   &
%   \ssumfs(\phrame : \cont)
%   &\wtsfs
%   \apushfs(\absmap(\phrame), \ssumfs(\cont))
%   \text.
% \end{align*}

\subsection{Example: A reachable-addresses summary}

\label{sec:reachable-address-summary}

The reachable-addresses summary is the set of all the addresses
directly touchable by a frame on the stack.
We formally define \emph{touch} through the touch function, $\touch_f
: \sa{Frame} \to \sa{Addr}$, which returns the addresses within the
given frame:
\begin{equation*}
  \touch_f(\vv,\expr,\aenv) =
  \setbuild{\aenv(\vv')}{\vv' \in \free(\expr) - \{\vv\}}
  \text.
\end{equation*}
The summary-space is the set of addresses:
\begin{equation*}
  \sa{Summary_{ra}} = \Pow{\sa{Addr}}
  \text.
\end{equation*}
The order on summaries is subset inclusion:
\begin{equation*}
  \ass \wtsra \ass'
  \text{ iff }
  \ass \subseteq \ass'
  \text.
\end{equation*}

The reachable address summarization function,
$\ssumra : \sa{Stack} \to \sa{Summary_{ra}}$,
finds the reachable addresses of each abstracted frame and keeps them in a set:
\begin{equation*}
  \ssumra\vect{\phrame_1,\ldots,\phrame_n} =
  \touch_f(\absmap_{Frame}(\phrame_1)) \union
  \cdots \union
  \touch_f(\absmap_{Frame}(\phrame_n))
  \text.
\end{equation*}
% \begin{equation*}
%   \ssumra(\cont) = \left\{
%   \begin{array}{ll}
%     \touch_f(\absmap(\phrame))
%     \cup
%     \ssumra(\cont')
%     &
%     \text{if } \cont = \phrame : \cont'
%     \\
%     \emptyset
%     &
%     \text{if } \cont \text{ is an empty stack} 
%   \end{array}
%   \right.
% \end{equation*}
%
% Every stack can be recursed down to the empty stack, so every stack has a representative summary.

The push operation $\apushra : \sa{Frame} \times \sa{Summary_{ra}} \to \sa{Summary_{ra}}$ adds the reachable addresses from the new frame to the set:
\begin{equation*}
  \apushra(\aphrame, \ass)
  =
  \touch_f(\aphrame)
  \cup
  \ass
  \text.
\end{equation*}
% % 
% The constraint on the push operation is true yet again by the various definitions:
% %
% \begin{align*}
%   \touch_f(\absmap(\phrame))
%   \cup
%   \ssumra(\cont)
%   &=
%   \ass
%   =
%   \ssumra(\phrame : \cont)
%   &
%   \ass
%   &\wtsra
%   \ass
%   \\
%   \touch_f(\absmap(\phrame))
%   \cup
%   \ssumra(\cont)
%   &=
%   \ass
%   =
%   \apushra(\absmap(\phrame), \ssumra(\cont))
%   &
%   \ssumra(\phrame : \cont)
%   &\wtsra
%   \apushra(\absmap(\phrame), \ssumra(\cont))
%   \text.
% \end{align*}

The reachable address summary provides the information about the stack needed for abstract garbage collection with pushdown control-flow analysis.
% %
% If an address is in the summary, then it is reachable from the stack.
% %
% Likewise, if an address is not in the summary, then it is not reachable from the stack (but it may be reachable from the control state).

% Local Variables: 
% TeX-master: "paper-lncs"
% End: 

\section{SSCFA: Stack-summarizing control-flow analysis} \label{sscfa}

In the last section, we defined stack summaries and motivated their
implementation informally.
In this section, we formally define the configuration-space and an
abstract pushdown semantics for stack-summarizing control-flow analysis
(SSCFA).
Appendix~\ref{sec:algorithm} describes a formal algorithm for creating
a finite model of the reachable state-space for SSCFA.

\subsection{Abstract configuration-space}
The only change between the configuration-spaces for the pushdown
control-flow analysis and the stack-summarizing control-flow analysis
is that configurations contain stack summaries instead of stacks:
\begin{align*}
  \aconf \in \sa{Conf} &= \sa{State} \times \sa{Summary} && \text{[configurations]}
  % \\
  %% \astate \in \sa{State} &= \syn{Exp} \times \sa{Env} \times \sa{Store} && \text{[states]}
  %% \\
  %% \aenv \in \sa{Env} &= \syn{Var} \parto \sa{Addr} && \text{[environments]}
  %% \\
  %% \astore \in \sa{Store} &= \sa{Addr} \to \Pow{\sa{Clo}} && \text{[stores]}
  %% % \\
  %% % \aden \in \sa{Den} &= \Pow{\sa{Clo}} && \text{[denotable values]}
  %% \\
  %% \aclo \in \sa{Clo} &= \syn{Lam} \times \sa{Env} && \text{[closures]}
  %% \\
  %% \acont \in \sa{Kont} &= \sa{Frame}^* && \text{[stacks]}
  %% \\
  %% \aphrame \in \sa{Frame} &= \syn{Var} \times \syn{Exp} \times \sa{Env} && \text{[stack frames]}
  %% \\
  %% \aaddr \in \sa{Addr} &\text{ is a \emph{finite} set of addresses} && \text{[addresses]}
  %% \\
  % \ass \in \sa{Summary} &\text{ is a finite set of stack summaries} && \text{[stack summaries]}
  \text.
\end{align*}

% Stack-summarizing control-flow analysis uses the same approach to
% configuration exploration as pushdown control-flow analysis uses.
% %
% In fact, we will redefine a few of the major functions and procedures
% in terms of the pushdown control-flow analysis versions.
% %
% \textbf{Dyck configuration graphs} ($\fDCG$) are Dyck state graphs
% where all the control states are replaced by configurations.
% %
% \ecg{} are likewise overloaded to use configurations instead of just
% control-states.
% %
% For space reasons, the formal definitions of these are omitted.
% %
% The top function, $\topfunction : \sa{Conf} \times \fDCG \times \fECG
% \to \Pow{\sa{Frame} \times \sa{Conf}}$, likewise works with
% configurations:
% %
% Formally, for a configuration $\aconf$, a Dyck configuration graph
% $\dcg$, and an \ecg{} $\epcg$:
% %
% \begin{equation*}
%   \topfunction(\aconf,\dcg,\epcg) = \setbuild{(\aphrame',\aconf')}{(\aconf'',\aconf) \in \epcg \text{ and } \aconf' \aoTo^{\aphrame_+} \aconf'' \in \dcg}
% \end{equation*}
%

\subsection{Abstract pushdown semantics}
The abstract transition relation for SSCFA is similar to the
transition relation for PDCFA.
The transition relation, $(\aToss) \subseteq \sa{Conf} \times
\sa{Frame}_\pm \times \sa{Conf}$ has three rules.
With respect to a program $\expr$, we can define a rooted pushdown
system, $M_{\rm SS} = (\sa{Conf},\sa{Frame},(\aToss),\aconf_0)$, where
$\aconf_0 = (\expr,[],[],\bot_{SS})$.

A tail call leaves the stack unchanged:
\begin{gather*}
  (\overbrace{(\sembr{\appform{\fexpr}{\aexpr}}, \aenv, \astore)}^{\astate}, \ass)\ 
  \aoToss^{\epsilon}\ 
  ((\expr,\aenv'',\astore'), \ass)
  \text{, where }
\\
\begin{align*}
  (\sembr{\lamform{\vv}{\expr}}, \aenv') &\in \aArgEval(\fexpr,\aenv,\astore)
  &
  \aenv'' &= \aenv'[\vv \mapsto \aaddr]
  \\
  \aaddr &= \aalloc(\vv,\astate)
  &
  \astore' &= \astore \join [\aaddr \mapsto \aArgEval(\aexpr,\aenv,\astore)]
  \text.
\end{align*}
\end{gather*}
A non-tail call builds a frame, adds it to the summary, and evaluates the call:
\begin{gather*}
  ((\sembr{\letiform{\vv}{\call}{\expr}}, \aenv, \astore), \ass)\
  \aoToss^{\aphrame_+}\ 
  ((\call,\aenv,\astore), \ass')
  \text{, where }
  \\  
  \begin{align*}
    \aphrame & = (\vv,\expr,\aenv)
    &
    \ass'  & = \apush(\aphrame, \ass)
    \text.
  \end{align*}
\end{gather*}
A function return pops the top frame off the stack.
It also restores older stack summaries.
Thus, the algorithm must know all of the abstract configurations on
paths from the initial configuration that can reach the current
configuration on a path whose net stack change is the the frame to be
popped; we find these abstract configurations using the $\apred :
\sa{Conf} \times \sa{Frame} \to \Pow{\sa{Conf}}$ function:
\begin{equation*}
  \apred(\aconf, \aphrame) = 
  \setbuild{ \aconf' }{ \aconf \mathrel{\RPDTrans_{M_{\rm SS}} ^{\vec{\phrame'}}} \aconf' \text{ and } 
  \fnet{\vec{\phrame'}} = \phrame_{+} } 
  \text.
\end{equation*}
The transition rule for pop is then straightforward:
\begin{gather*}
  (\overbrace{(\aexpr, \aenv, \astore)}^{\astate}, \ass)
  \mathrel{\aoToss^{\aphrame_-}}
  ((\expr,\aenv'',\astore'), \ass')
  \text{, where }
  \\
  \begin{align*}
    (\_,\ass') &\in \apred(\aconf,\aphrame)
    %(\aphrame, \aconf'') &\in  \topfunction(\state,\dsg,\epcg)
    &
    \aenv'' &= \aenv'[\vv \mapsto \aaddr]
    \\
    (\vv,\expr,\aenv') &= \aphrame
    &
    \astore' &= \astore \join [\aaddr \mapsto \aArgEval(\aexpr,\aenv,\astore)]
    \\
%    ((\expr', \aenv''', \astore''), \ass') &= \aconf''
%    &
    \aaddr &= \aalloc(\vv,\astate)
    \text.  
  \end{align*}
\end{gather*}

\subsection{Soundness of stack-summarizing control-flow analysis}

The soundness of Dyck state graphs has been proved
in~\cite{local:Earl:2010:PDCFA}.
However, the soundness of the stack summaries is provided below for the first time:

\begin{theorem} \label{soundness}
  If $\absmap(\conf) \wt \aconf$, $\conf \To \conf'$
  and $\conf_0 \To^* \conf$,
  then there exists $\aconf' \in \sa{Conf}$ such that
  $\absmap(\conf') \wt \aconf'$ and $\aconf \aToss \aconf'$.
\end{theorem}
\begin{proof}[sketch]
Let $\conf = (\state, \cont)$,
$\conf' = (\state', \cont')$ and
$\aconf = (\astate, \ass)$,
such that
$\absmap(\conf) \wt \aconf$.
We know by theorems in~\cite{local:Earl:2010:PDCFA} that there exists a state
$\astate'' \in \sa{State}$
such that
$\absmap(\state') \wt \astate''$.
We also know that the first stack is subsumed by the first stack summary:
$\ssum(\cont) \wts \ass$.
So we must prove that there exists a stack summary
$\ass' \in \sa{Summary}$
such that
$\ssum(\cont') \wts \ass'$
and
$\aconf \aToss (\astate'', \ass')$.
The proof continues with a case-wise analysis on the type of the transition as well as strong induction based upon the length of the path to configuration $\conf$.  See Appendix~\ref{proofs} for details.
\end{proof}
%

% Local Variables: 
% TeX-master: "paper-lncs"
% End: 

\section{SSCFA with Abstract Garbage Collection}
\label{ssgammacfa}

Having constructed a framework for iteratively synthesizing stack
summaries during computation of a finite model for a pushdown system,
we can integrate abstract garbage collection.
In this section, we assume the ``reachable addresses'' stack summary
is in use.
We term this analysis SS$\Upgamma$CFA, the product of stack-summarizing
control-flow analysis and abstract garbage collection (also called
$\Upgamma$CFA).
SS$\Upgamma$CFA is a ``best of both worlds'' combination: it has all
the argument precision advantages of abstract garbage collection
and all the return-flow precision advantages of PDCFA.

As with classical abstract garbage collection, we must define what
makes an address or value reachable.
Essentially, an object is \emph{reachable} if it may be used either in
the current configuration or in a subsequent configuration.
If an address is reachable, all the values bound to it are also
reachable.
Values (closures and frames) reference addresses through their
environments, which are reachable as well.
Because values touch addresses and addresses touch values, finding
reachable addresses and values amounts to a bipartite graph search.
The concrete values of unreachable addresses will never be used again
during the course of the computation; thus, it is safe to set the
values of these addresses to bottom within the store.

The reachability exploration of the store begins with the addresses
that the current configuration $\aconf$ can immediately reach, called
the root set, $\rootset(\aconf)$.
The root function, $\rootset : \sa{Conf} \to \Pow{\sa{Addr}}$,
returns the root set for a configuration:
\[
\rootset((\expr,\aenv,\astore),\ass) = 
\ass \cup \setbuild{\aenv(\vv)}{\vv \in \free(\expr)}
\text,
\]
where the function $\free : \syn{Exp} \to \Pow{\syn{Var}}$ returns the free variables in the given expression.
The root set contains all the addresses bound to free variables in the expression, $\expr$, as well as the addresses in the reachable address summary.

The touch function, $\touch_c : \sa{Clo} \to \sa{Addr}$, finds addresses referenced in closures:
\(
  \touch_c(\lam,\aenv) = 
  \setbuild{\aenv(\vv)}{\vv \in \free(\lam)}
\).
The touching relation, $(\touchrel{\astore}) : \sa{Addr} \to \sa{Addr}$ links addresses directly to addresses:
\[
  \aaddr \touchrel{\astore} \aaddr' \text{ iff }
  \aaddr' \in \touch_c(\aval) \text{ and }
  \aval \in \astore(\aaddr)
  \text.
\]
With this relation, finding all reachable addresses of a configuration
$\aconf$ becomes the transitive closure of the touching relation:
\begin{equation*}
  \reach(\aconf) = \setbuild{\aaddr'}{\aaddr \touchrel{\astore}^* \aaddr' \text{ and }
    \aaddr \in \rootset(\aconf)}
  \text.
\end{equation*}
Finally, we define the abstract garbage collector itself, $\agc :
\sa{Conf} \to \sa{Conf}$, which simply restricts the store to the
reachable addresses:\footnote{
We define function restriction, $f|X$, so that $f|X = \lambda x . \cif{x \in X}{f(x)}{\bot}$.
}
\begin{equation*}
  \agc(\aconf) = (\expr,\aenv,\astore\vert\reach(\aconf),\ass)
  \text{, where } \aconf = (\expr,\aenv,\astore,\ass)\text.
\end{equation*}

The abstract transition relation for SS$\Upgamma$CFA, $(\aToagc) \subseteq
\sa{Conf} \times \sa{Conf}$, needed for stack-summarizing control-flow
analysis with abstract garbage collection extends the abstract
transition relation to collect before each transition:
\begin{equation*}
  \aconf \aToagc \aconf'
  \text{ iff }
  \agc(\aconf) \aToss \aconf'
  \text.
\end{equation*}
%
% The proof of Theorem \ref{soundness} holds for this analysis because
% the only parameter (besides the abstract allocation function) is the
% push function, which meets the constraint for valid push functions.
%
The soundness theorems and their proofs for classical abstract garbage
collection are in Chapter 6 of \cite{mattmight:Might:2007:Dissertation};
they adapt readily to our pushdown framework.

% Local Variables: 
% TeX-master: "paper-sigplan"
% End: 

%\section{Results} \label{results}

\section{Related Work} \label{related}

Stack summarization, the central contribution of this paper, overcomes
the apparent incompatibilities of two orthogonal anti-merging
techniques designed to improve precision: abstract garbage
collection~\cite{mattmight:Might:2006:GammaCFA} and pushdown
control-flow
analysis~\cite{local:Earl:2010:PDCFA,mattmight:Vardoulakis:2010:CFA2}.
As such, this work directly builds upon both techniques, as well as
classical control-flow analysis~\cite{mattmight:Shivers:1991:CFA},
abstract machines~\cite{mattmight:Felleisen:1987:CESK}, and abstract
interpretation~\cite{mattmight:Cousot:1977:AI,mattmight:Cousot:1979:Galois}
in general.

Abstract garbage collection~\cite{mattmight:Might:2006:GammaCFA,local:VanHorn:2010:Abstract} curbs
argument-merging, but it has not yet been applied to anything
beyond classical control-flow analysis.
%
%% Abstract garbage collection cannot address return-flow merging, except
%% for the notable exception of tail calls: abstract garbage collection
%% can collect continuations between invocations of tail calls.

% Pushdown control-flow analysis~\cite{local:Earl:2010:PDCFA} provides a
% framework where stack summaries can easily be added, recovered, and
% used so that local properties, such as reachable addresses can be
% retained.
% %
% None of the work related to pushdown control-flow analysis below
% includes stack summaries.
% %
% Theoretically any of the work below can be expanded to include stack
% summaries.

Vardoulakis and Shivers's CFA2~\cite{mattmight:Vardoulakis:2010:CFA2}
is the precursor to the pushdown control-flow
analysis~\cite{local:Earl:2010:PDCFA} presented in
Section~\ref{pdcfa}.
CFA2 is a table-driven summarization algorithm that exploits the
balanced nature of calls and returns to improve return-flow precision
in a control-flow analysis.
While CFA2 uses a concept called ``summarization,'' it is a
summarization of execution paths of the analysis, roughly equivalent
to Dyck state graphs rather than our stack summaries.

In terms of recovering precision, pushdown control-flow
analysis~\cite{local:Earl:2010:PDCFA} is the dual to abstract garbage
collection:
it focuses on the global interactions of configurations via
transitions to precisely match push-pop/call-return, thereby
eliminating all return-flow merging.
However, pushdown control-flow analysis does nothing to improve
argument merging.

In the context of first-order languages, pushdown approaches to 
analysis are well-established.
Reps \emph{et al.}~\cite{mattmight:Reps:1995:Precise} uses a
summarization algorithm to compute a Dyck-state-graph-like solution.
Debray and Proebsting~\cite{dvanhorn:Debray1997Interprocedural}
develop an analysis with perfect return-flow in the presence of tail
calls.
For higher-order languages, finite-state approaches
\emph{approximating} the pushdown precision of return-flow have been
explored by Midtgaard and Jensen~\cite{mattmight:Midtgaard:2009:CFA}
and Van Horn and Might~\cite{local:VanHorn:2010:Abstract}.
Our work extends the pushdown approach to higher-order languages with
tail calls, and produces stack summaries to enable abstract garbage
collection.

%% or higher-order functions---both principles
%% objectives of our own work.
%% %
%% Again, that work does not contain stack summaries.

%% Kobayashi's recent work~\cite{mattmight:Kobayashi:2009:HORS} verifies
%% properties of higher-order programs using model checking of
%% higher-order recursion schemes.
%% %
%% The recursion schemes generate infinite trees, which are checked for
%% the desired properties.
%% %
%% Kobayashi's framework is both powerful and general, and it seems
%% likely that we could adapt abstract garbage collection to this work as
%% a next step.

%% Since abstract garbage collection has not been converted to use in
%% model-checking, SSCFA and SS$\Gamma$CFA are orthogonal to Kobayashi's
%% work.

% Local Variables: 
% TeX-master: "paper-sigplan"
% End: 

\section{Conclusion} \label{concl}

% We presented stack-summarizing CFA, an algorithm for computing a
% pushdown-like CFA with access to finite knowledge of the stack
% \emph{during} the analysis.
% %
% We also presented an abstract-garbage-collecting instantiation of
% SSCFA, which uses stack summaries to track stack-reachable addresses.
%
We presented SS$\Upgamma$CFA, a synergistic fusion of pushdown
analysis and abstract garbage collection to combat the twin sinks for
precision in higher-order flow analysis: merging in arguments, and
merging in return-flow.
In order to create SS$\Upgamma$CFA, we had to first create SSCFA, a
pushdown control-flow analysis for higher-order programs capable of
iteratively synthesizing summaries of stack properties; in this case,
we required a summary of reachable addresses on the stack.
Abstract garbage collection combats merging in arguments by
eliminating monotonicity for the abstract store; pushdown analysis
eliminates the loss in return-flow precision by simulating the
concrete call stack with a pushdown stack, thereby properly matching
returns to call.

% Local Variables: 
% TeX-master: "paper-sigplan"
% End: 

% \section{Bibliography}

%%% Local Variables: 
%%% mode: latex
%%% TeX-master: "paper-lncs.tex"
%%% End: 

% LocalWords:  expr Exp TODO iff prelim stackify CFA SSCFA PDCFA PDS RPDS Configs PDA aw Dyck CESK Clo Env Conf polyvariance Var injector op

\bibliographystyle{plain}
\bibliography{bibliography,local} 
\newpage
\appendix

\section{Classical control-flow analysis} \label{classical-cfa}

This section presents traditional control-flow analysis for reference
and comparison with pushdown and stack-summarizing control flow analysis.
Classical control-flow analysis for ANF operates over the abstract
state-space in Figure \ref{fig:abs-conf-space}.
Our classical formulation follows Van Horn and Might's technique of
allocating abstract continuations in the store, as opposed to stacking
them~\cite{local:VanHorn:2010:Abstract}.

\begin{figure}
\begin{align*}
  \aconf \in \sa{Conf} &= \sa{State} \times \sa{Addr} && \text{[configurations]}
  \\
  \astate \in \sa{State} &= \syn{Exp} \times \sa{Env} \times \sa{Store} && \text{[states]}
  \\
  \aenv \in \sa{Env} &= \syn{Var} \parto \sa{Addr} && \text{[environments]}
  \\
  \astore \in \sa{Store} &= \sa{Addr} \to \Pow{\sa{Clo} \cup \sa{Frame}} && \text{[stores]}
  % \\
  % \aden \in \sa{Den} &= \Pow{\sa{Clo}} && \text{[denotable values]}
  \\
  \aclo \in \sa{Clo} &= \syn{Lam} \times \sa{Env} && \text{[closures]}
  \\
  \aphrame \in \sa{Frame} &= \syn{Var} \times \syn{Exp} \times \sa{Env} \times \sa{Addr} && \text{[stack frames]}
  \\
  \arp, \aaddr \in \sa{Addr} &\text{ is a \emph{finite} set of addresses} && \text{[addresses]}
  %\text.
\end{align*}
\caption{Abstract configuration-space for classical control-flow analysis.}
\label{fig:abs-conf-space}
\end{figure}

To complete the abstract semantics we need to define program injection, atomic expression evaluation, reachable configurations, transition relation, address allocation, and abstraction function:

\paragraph{Program injection}
The abstract injection function $\aInject : \syn{Exp} \to \sa{Conf}$
pairs an expression with an empty environment, an empty store and an
empty stack to create the initial abstract configuration:
\begin{equation*}
  \aconf_0 = \aInject(\expr) = (\expr, [], [], \nulladdr)
  \text,
\end{equation*}
where $\nulladdr$ is an address bound to nothing in the store, thus representing an empty stack.

\paragraph{Atomic expression evaluation}
The abstract atomic expression evaluator, $\aArgEval : \syn{Atom}
\times \sa{Env} \times \sa{Store} \to \PowSm{\sa{Clo} \cup \sa{Frame}}$, returns the value of
an atomic expression or a stack frame in the context of an environment and a store;
note how it returns a set:
\begin{align*}
  \aArgEval(\lam,\aenv,\astore) &= \set{(\lam,\env)} && \text{[closure creation]}
  \\
  \aArgEval(\vv,\aenv,\astore) &= \astore(\aenv(\vv)) && \text{[variable look-up]}
  \text.
\end{align*}

\paragraph{Reachable configurations}
The abstract program evaluator $\aEval : \syn{Exp} \to
\PowSm{\sa{Conf}}$ returns all of the configurations reachable from
the initial configuration:
\begin{equation*}
  \aEval(\expr) = \setbuild{ \aconf }{ \aInject(\expr) \aTo^* \aconf } 
  \text.
\end{equation*}

\paragraph{Transition relation}
The abstract transition relation $(\aTo) \subseteq \sa{Conf} \times
\sa{Conf}$ has three rules, two of which have become nondeterministic.
A tail call may fork because there could be multiple abstract closures
that it is invoking:
\begin{align*}
  (\overbrace{(\sembr{\appform{\fexpr}{\aexpr}}, \aenv, \astore)}^{\astate}, \arp)
  &\aTo
  ((\expr,\aenv'',\astore'),\arp)
  \text{, where }
  \\
  (\sembr{\lamform{\vv}{\expr}}, \aenv') &\in \aArgEval(\fexpr,\aenv,\astore)
  \\
  \aaddr &= \aalloc(\vv,\astate)
  \\
  \aenv'' &= \aenv'[\vv \mapsto \aaddr]
  \\
  \astore' &= \astore \join [\aaddr \mapsto \aArgEval(\aexpr,\aenv,\astore)]
  \text.
\end{align*}
The partial order for stores is:
\begin{equation*}
  (\astore \join \astore')(\aaddr) = \astore(\aaddr) \union \astore'(\aaddr)
  \text.
\end{equation*}

\noindent
A non-tail call builds a frame, adds it to the store, and evaluates the call:
\begin{align*}
  (\overbrace{(\sembr{\letiform{\vv}{\call}{\expr}}, \aenv, \astore)}^{\astate}, \arp)
  &\aTo
  ((\call,\aenv,\astore), \arp')
  \text{, where }
  \\
  \arp' &= \aalloc(\vv,\astate)
  \\
  \astore' &= \astore \join [\arp' \mapsto (\vv,\expr,\aenv,\arp)]
  \text.
\end{align*}

\noindent
A function return may fork because there could be multiple frames bound to the current return pointer:
\begin{align*}
  (\overbrace{(\aexpr, \aenv, \astore)}^{\astate}, \arp)
  &\aTo
  ((\expr,\aenv'',\astore'), \arp')
  \text{, where }
  \\
  (\vv,\expr,\aenv',\arp') &\in \astore(\arp)
  \\
  \aaddr &= \aalloc(\vv,\astate)
  \\
  \aenv'' &= \aenv'[\vv \mapsto \aaddr]
  \\
  \astore' &= \astore \join [\aaddr \mapsto \aArgEval(\aexpr,\aenv,\astore)]
  \text.  
\end{align*}

\paragraph{Allocation, polyvariance and context-sensitivity}
\label{sec:polyvariance}
In the abstract semantics, the abstract allocation function
$\aalloc : \syn{Var} \times \sa{State} \to \sa{Addr}$ determines the
polyvariance of the analysis (and, by extension, its
context-sensitivity).
The abstract allocation function is overloaded to assign return pointers
(addresses) to abstract stack frames:
$\aalloc : \sa{Frame} \times \sa{State} \to \sa{Addr}$.
In a control-flow analysis, \emph{polyvariance} literally refers to
the number of abstract addresses (variants) there are for each
variable.
% %
% Matt> Cut and paste error?
%
% By selecting the right abstract allocation function, we can
% instantiate pushdown versions of classical flow analyses.

\paragraph{Abstraction function}

% Matt> Can't define one easily with respect to the concrete semantics.
%       Needs VH&M transformations on the CESK machine first.

The abstraction function ($\absmap$) converts any structure from the
concrete semantics (Figure \ref{fig:conc-abs-conf-space}) into an abstract
form of the same structure (Figure \ref{fig:abs-conf-space}).
(The abstraction function is defined in Appendix~\ref{abstraction-functions}.  While not specifically defined for these semantics, the abstraction function there can easily be modified to work with return pointers.)

\paragraph{Comparison to pushdown control-flow analysis}

The abstract semantics of pushdown control-flow analysis 
% (Figure \ref{fig:abs-pd-state-space}) 
are similar to those of the abstract
semantics for classical control-flow analysis (Figure
\ref{fig:abs-conf-space}).
However, the few key differences are worth noting.

Foremost, we are not working with the configuration-space directly;
rather we deal with the control-state-space.
Hence, a configuration is now defined as a state and a stack paired
together.
The results of pushdown control-flow analysis are rooted pushdown
systems rather than nondeterminisitic finite automata.
A rooted pushdown system handles the stack and configurations
implicitly, so we use the control-state-space instead of the
configuration-space.

The next change is that the store only contains bindings from
addresses to closures, instead of from addresses to closures and
frames.
The abstraction of the store creates imprecision.
By keeping the frames out of the store (and in a precise structure) we
avoid this imprecision for continuations.

The final difference is that frames no longer contain return pointers.
Again, this is because the enriched abstract transition system
encapsulates that information precisely.

% The abstraction function converts a ... TODO!

% Local Variables: 
% TeX-master: "paper-sigplan"
% End: 

\section{Abstraction Functions} \label{abstraction-functions}

This section defines the abstraction functions used throughout this paper.
In particular, the following abstraction functions links the concrete and the abstract configuration-spaces in Section~\ref{pdcfa}.
The abstraction function recurs structurally:
\begin{align*}
  \absmap(\state, \cont)
  &= (\absmap_{State}(\state), \absmap_{Kont}(\cont))
  && \text{[configuration abstraction]}
  \\
  \absmap_{State}(\expr, \env, \store)
  &= (\expr, \absmap_{Env}(\env), \absmap_{Store}(\store))
  && \text{[state abstraction]}
  \\
  \absmap_{Env}(\env)(\vv)
  &= \absmap_{Addr}(\env(\vv))
  && \text{[environment abstraction]}
  \\
  \absmap_{Store}(\store)(\aaddr)
  &= \bigsqcup_{\absmap_{Addr}(\addr) = \aaddr} \absmap_{Clo}(\store(\addr))
  && \text{[store abstraction]}
  \\
  \absmap_{Clo}(\lam, \env)
  &= \set{(\lam, \absmap_{Env}(\env))}
  && \text{[closure abstraction]}
  \\
  \absmap_{Kont}(\vect{\phrame_1, \dots, \phrame_n})
  &= \vect{\absmap_{Frame}(\phrame_1), \dots, \absmap_{Frame}(\phrame_n)}
  && \text{[stack abstraction]}
  \\
  \absmap_{Frame}(\vv, \expr, \env)
  &= (\vv, \expr, \absmap_{Env}(\env))
  && \text{[frame abstraction]}
  \text.
\end{align*}

Just as address-allocation is a parameter, the address abstraction function, $\absmap_{Addr} : \Addr \to \aAddr$, is a parameter for the abstract semantics.

For Sections~\ref{ssintro}, ~\ref{sscfa}, and~\ref{ssgammacfa}, the abstraction function for configurations is:
\begin{align*}
  \absmap(\state, \cont)
  &= (\absmap_{State}(\state), \ssum(\cont))
  && \text{[configuration abstraction]}
  \text,
\end{align*}
where the stack summarization function, $\ssum$, is a parameter as described in Section~\ref{ssintro}.

% The abstraction function converts a ... TODO!

% Local Variables: 
% TeX-master: "paper-sigplan"
% End: 

\section{Building a Dyck configuration graph}

\label{sec:algorithm}

\subsection{Building a Dyck state graph for PDCFA}

A fixed-point approach to building Dyck state graphs for PDCFA is best presented in~\cite{local:Earl:2010:PDCFA}.
The algorithm in Figure \ref{fig:build-DSG-algorithm} is a similar iterative algorithm, but it is formulated for stack-summarizing control-flow analysis (Section \ref{sscfa}).
The underlying approaches are similar.
In fact, for the algorithm in Figure \ref{fig:build-DSG-algorithm},
replacing configurations with states and switching the transition relation used throughout to the transition relation of Section~\ref{pdcfa} ($\aoTons$) is enough to convert the algorithm to build standard Dyck state graphs.
The main difference is that the algorithm presented here examines a single state, transition, or shortcut edge each iteration, whereas the fixed-point algorithm examines the entire frontier each iteration.

\subsection{Building a Dyck state graph for SSCFA}

The algorithm in Figure~\ref{fig:build-DSG-algorithm} builds the Dyck configuration graph and the \ecg{} for a given program $\expr$.
It uses the worklists, $\delconf$, $\deledge$, and $\delshort$, to maintain the frontier of unexplored configurations, transitions, and shortcut edges respectively.
The while loop runs until all the worklists are empty, which is exactly when everything reachable has been explored.
Each iteration of the while loop, explores one previously unexplored shortcut edge, transition, or configuration.

%\paragraph{Syntactic sugar:}
%
%% When a triple $(x,\ell,x')$ is an edge in a labeled graph, we write
%% $x \pdedge^\ell x'$.
%% %
%% Similarly, when a pair $(x,x')$ is a graph edge, we write
%% $\biedge{x}{x'}$.
%% %
%% We use both string and vector notation for sequences:
%% $a_1 a_2 \ldots a_n \equiv \vect{a_1,a_2,\ldots,a_n} \equiv \vec{a}$.

Each new configuration, transition, and shortcut edge can imply other configurations, transitions, and shortcut edges globally.
Thus, after each configuration, transition, or shortcut edge is explored, the implied configurations, transitions, and shortcut edges are added to the worklists.

The procedure $\addEmpty$ finds all the configurations, transitions, and shortcut edges implied by the given shortcut edge.
%
%Figure \ref{fig:shortcut-implies} diagrams all the shortcut edges that can be implied by a shortcut edge.
%
The only new transitions implied by a shortcut edge are pop transitions that are enabled by a new push transition.
The only new configurations are those in the newly implied pop transitions.

The procedure $\addEdge$ finds all the configurations, transitions, and shortcut edges implied by the given transition.
There are three types of transitions:
First, there are no-op ($\epsilon$) transitions, which immediately become shortcut edges, and so are added to the \ecg{} and are expanded.
Next, there are push transitions, which imply new pop transitions (and configurations from these pop transitions) as well as new shortcut edges (from pre-existing pop transitions).
Finally, there are pop transitions, which imply only new new shortcut edges from pre-existing push transitions.

The last procedure $\Explore$ finds all the configurations, transitions, and shortcut edges implied by the given configuration.
A configuration cannot imply any shortcut edges directly.
However, a configuration can imply new no-op, push, or pop transitions as well as new configurations from these transitions.

\begin{figure}
{
\def\ind{\hspace{0.2in}}
\def\gets{\mathrel{\leftarrow}}
\parindent=0.0in

\begin{tabular}{l@{\hspace{0.1in}}l}
  \textbf{procedure} $\textit{BuildDyck}(\expr)$
  \\
  $\ind\aconf_0 \gets \aInject_\aconf(\expr)$;
  $\dsg \gets (\emptyset,\StackAlpha,\emptyset,\aconf_0)$;
  $\epcg \gets (\emptyset,\emptyset)$;
  $\delconf \gets \{\aconf_0\}$;
  $\deledge \gets \emptyset$;
  $\delshort \gets \emptyset$
  \\
  $\ind$\textbf{while}
  $(
  \delconf \neq \emptyset
  \textbf{ or }
  \deledge \neq \emptyset
  \textbf{ or }
  \delshort \neq \emptyset
  )$
  \\
  $\ind$
  $\ind$
  \textbf{if} $(\delshort \neq \emptyset),$\textbf{ let}
  $(\aconf,\aconf') \in \delshort$
  \\
  $\ind$
  $\ind$
  $\ind$
  $(\delconf',\deledge',\delshort') \gets \addEmpty(\dsg,\epcg)(\biedge{\aconf}{\aconf'})$
  \\
  $\ind$
  $\ind$
  $\ind$
  $(\DSConfs, \ECEdges) \gets \epcg$
  \\
  $\ind$
  $\ind$
  $\ind$
  $\epcg \gets (\DSConfs, \ECEdges \cup \{(\aconf,\aconf')\})$
  \\
  $\ind$
  $\ind$
  $\ind$
  $(\delconf,\deledge,\delshort) \gets 
  (\delconf \cup \delconf',
  \deledge \cup \deledge',
  \delshort \cup \delshort' - \{(\aconf,\aconf')\})$
  \\
  $\ind$
  $\ind$
  \textbf{else if} $(\deledge \neq \emptyset),$\textbf{ let}
  $(\aconf,\stackact,\aconf') \in \deledge$
  \\
  $\ind$
  $\ind$
  $\ind$
  $(\delconf',\deledge',\delshort') \gets \addEdge(\dsg,\epcg)(\aconf \pdedge^\stackact \aconf')$
  \\
  $\ind$
  $\ind$
  $\ind$
  $(\DSConfs, \StackAlpha, \DSEdges, \aconf_0) \gets \dsg$
  \\
  $\ind$
  $\ind$
  $\ind$
  $\dsg \gets (\DSConfs, \StackAlpha, \DSEdges \cup \{(\aconf,\stackact,\aconf')\},\aconf_0)$
  \\
  $\ind$
  $\ind$
  $\ind$
  $(\delconf,\deledge,\delshort) \gets 
  (\delconf \cup \delconf',
  \deledge \cup \deledge' - \{(\aconf,\stackact,\aconf')\},
  \delshort \cup \delshort')$
  \\
  $\ind$
  $\ind$
  \textbf{else if} $(\delconf \neq \emptyset),$\textbf{ let}
  $\aconf \in \delconf$
  \\
  $\ind$
  $\ind$
  $\ind$
  $(\delconf',\deledge',\delshort') \gets \Explore(\dsg,\epcg)(\aconf)$
  \\
  $\ind$
  $\ind$
  $\ind$
  $(\DSConfs, \StackAlpha, \DSEdges, \aconf_0) \gets \dsg$
  \\
  $\ind$
  $\ind$
  $\ind$
  $(\DSConfs, \ECEdges) \gets \epcg$
  \\
  $\ind$
  $\ind$
  $\ind$
  $(\dsg,\epcg) \gets 
  ((\DSConfs \cup \{\aconf\}, \StackAlpha, \DSEdges,\aconf_0),
  (\DSConfs \cup \{\aconf\}, \ECEdges))$
  \\
  $\ind$
  $\ind$
  $\ind$
  $(\delconf,\deledge,\delshort) \gets 
  (\delconf \cup \delconf' - \{\aconf\},
  \deledge \cup \deledge',
  \delshort \cup \delshort')$
  \\
  $\ind$\textbf{return} $\dsg,\epcg$
\end{tabular}

\vspace{.5em}

\begin{tabular}{l@{\hspace{0.1in}}l}
  \textbf{procedure} $\addEmpty(\dsg,\epcg)({\aconf},{\aconf'})$
  \\
  $\ind(\DSConfs, \StackAlpha, \DSEdges, \aconf_0) \gets \dsg$;
  $(\DSConfs, \ECEdges) \gets \epcg$
  \\
  $\ind\deledge \gets \setbuild{(\aconf',{\aphrame_-},\aconf_2)}
    {
      (\aconf_1,{\aphrame_+},\aconf) \in \DSEdges
      \text{ and }
      \aconf' \aoToss^{\aphrame_-} \aconf_2
    }$
  \\
  $\ind\delconf \gets \setbuild{\aconf_2}{(\aconf_1,{\aphrame_-},\aconf_2) \in \deledge}$
  \\
  $\ind\delshort \gets \setbuild{({\aconf_1},{\aconf'})}{({\aconf_1},{\aconf}) \in \ECEdges}$
  \\
  $\ind\ind\ind\cup
  \setbuild{({\aconf},{\aconf_2})}{({\aconf'},{\aconf_2}) \in \ECEdges}$
  \\
  $\ind\ind\ind\cup
  \setbuild{({\aconf_1},{\aconf_2})}
    {
      ({\aconf_1},{\aconf}), 
      ({\aconf'},{\aconf_2}) \in \ECEdges
    }$
  \\
  $\ind$\textbf{return} $\delconf - \DSConfs,\ \deledge - \DSEdges,\ \delshort - \ECEdges$
\end{tabular}

\vspace{.5em}

\begin{tabular}{l@{\hspace{0.1in}}l}
  \textbf{procedure} $\addEdge(\dsg,\epcg)(\aconf \pdedge^\stackact \aconf')$
  \\
  $\ind(\DSConfs, \StackAlpha, \DSEdges, \aconf_0) \gets \dsg$;
  $(\DSConfs, \ECEdges) \gets \epcg$
  \\
  $\ind$\textbf{if} $(\stackact = \epsilon)$
  \textbf{return} $\addEmpty(\dsg,(\DSConfs, \ECEdges \cup \{({\aconf},{\aconf'})\}))({\aconf},{\aconf'})$
  \\
  $\ind$\textbf{else if} $(\stackact = \aphrame_+)$
  \\
  $\ind\ind$
  $\deledge \gets \setbuild{\aconf_1 \pdedge^{\aphrame_-} \aconf_2}
    {
      \biedge{\aconf'}{\aconf_1} \in \ECEdges
      \text{ and }
      \aconf_1 \aoToss^{\aphrame_-} \aconf_2
    }$
  \\
  $\ind\ind$
  $\delconf \gets \setbuild{\aconf_2}{\aconf_1 \pdedge^{\aphrame_-} \aconf_2 \in \deledge}$
  \\
  $\ind\ind$
  $\delshort \gets \setbuild{\biedge{\aconf}{\aconf_2}}
    {
      \biedge{\aconf'}{\aconf_1} \in \ECEdges
      \text{ and }
      \aconf_1 \pdedge^{\aphrame_-} \aconf_2 \in \DSEdges
    }$
  \\
  $\ind\ind$
  \textbf{return} $\delconf - \DSConfs,\ \deledge - \DSEdges,\ \delshort - \ECEdges$
  \\
  $\ind$\textbf{else if} $(\stackact = \aphrame_-)$
  \\
  $\ind\ind$
  \textbf{return} $\emptyset,\ \emptyset,\ 
  \setbuild{\biedge{\aconf_1}{\aconf'}}
           {
             \biedge{\aconf_2}{\aconf} \in \ECEdges
             \text{ and }
             \aconf_1 \pdedge^{\aphrame_+} \aconf_2 \in \DSEdges
           }- \ECEdges$
\end{tabular}

\vspace{.5em}

\begin{tabular}{l@{\hspace{0.1in}}l}
  \textbf{procedure} $\Explore(\dsg,\epcg)(\aconf)$
  \\
  $\ind (\DSConfs, \StackAlpha, \DSEdges, \aconf_0) \gets \dsg$;
  $(\DSConfs, \ECEdges) \gets \epcg$
  \\
  %% $\delshort \gets \emptyset$
  %% \\
  $\ind$\textbf{return} 
  $\setbuild{\aconf'}
    {\aconf \aoToss^{\stackact} \aconf'} - \DSConfs,\ 
    \setbuild{(\aconf,\stackact,\aconf')}
             {\aconf \aoToss^{\stackact} \aconf'} - \DSEdges,\ 
             \emptyset$
\end{tabular}
}
\caption{Algorithm to build a Dyck configuration graph. 
  Procedures $\addEmpty$, $\addEdge$ and $\Explore$ determine
  what configurations, transitions, and shortcut edges are implied by
  a given shortcut edge, transition, or configuration, respectively.}
\label{fig:build-DSG-algorithm}
\end{figure}

\subsection{Building a Dyck state graph for SS$\Gamma$CFA}

Finally, we let the procedures of Figure~\ref{fig:build-DSG-algorithm} use this transition relation ($\aToagc$) and the reachable address push function $\apushra$.
Now the procedure $\textsc{BuildDyck}$ of Figure \ref{fig:build-DSG-algorithm} computes stack-summarizing control-flow analysis with abstract garbage collection soundly.

\section{Proofs}
\label{proofs}

\begin{proof}[of Theorem~\ref{soundness}]
Without loss of generality, assume that the path from the initial configuration $\conf_0$ to $\conf$ is length $n$, so the path to the new configuration $\conf'$ is $n+1$ transitions from the initial configuration.
The inductive hypothesis is that the theorem holds for all paths of length less than or equal to $n$.
So far the scenario can be diagrammed as below:
\[
\xymatrix{
  \conf_0 \ar@{=>}[r] \ar[d]_\absmap
  &
  \dots \ar@{=>}[r]
  &
  \conf \ar@{=>}[r] \ar[d]_\absmap
  &
  \conf' \ar[d]_\absmap
  \\
  \aconf_0 \ar@2{~>}[r]
  &
  \dots \ar@2{~>}[r]
  &
  \aconf \ar@2{~>}[r]^{\stackact}
  &
  \aconf'
}
\]
We now have three cases depending on what $\stackact$ is:
\begin{itemize}
  \item $\aconf \aToss^\epsilon \aconf' = (\astate'', \ass')$
    \\
    No change is made to the stack in the concrete, thus the stacks are equal:
    $\cont = \cont'$.
    Likewise, no change is made in the abstract, thus:
    $\ass = \ass'$.
    Since the first stack is subsumed by the first stack summary,
    the second stack must be subsumed by the second stack summary:
    $\ssum(\cont') \wts \ass'$.
  \item  $\aconf \aToss^{\aphrame_+} \aconf' = (\astate'', \ass')$
    \\
    A frame $\phrame$, such that $\absmap(\phrame) \wt \aphrame$, must be pushed onto the stack in the concrete, so the stacks are related thusly:
    $\phrame : \cont = \cont'$.
    Likewise, the stack summaries are so related:
    $\apush(\aphrame, \ass) \wts \ass'$.
    By the constraint on all push operations:
    $\ssum(\phrame : \cont) \wts \apush(\aphrame, \ssum(\cont))$.
    We can make the following replacements:
    $\ssum(\cont') \wts \apush(\aphrame, \ass)$.
    By the definition of the transitive relation ($\aToss$):
    $\ass' = \apush(\aphrame, \ass)$.
    Finally we have:
    $\ssum(\cont') \wts \ass'$.
  \item  $\aconf \aToss^{\aphrame_-} \aconf' = (\astate'', \ass')$
    \\
    A frame $\phrame$, such that $\absmap(\phrame) \wt \aphrame$, must be popped from the stack in the concrete, so the stacks are related like this:
    $\cont = \phrame : \cont'$.
    We know that configuration $\conf$ is reachable from the initial configuration, which has an empty stack.
    Since the transition from the configuration pops off a frame $\phrame$, it does not currently have an empty stack.
    So there exists a path,
    $\conf_0 \To^* \conf_1 \To^{\phrame_+} \conf_2 \To^{\vec{\stackact}} \conf$,
    such that the net of the stack actions after the push, $\fnet{\vec{\stackact}}$, is empty.
    Let
    $\conf_1 = (\state_1, \cont_1)$.
    Since the net of the stack actions is empty, the stack at the configuration before the last previously unmatched push, $\cont_1$, is identical to the stack after the current pop, $\cont'$:
    $\cont_1 = \cont'$.
    
    By the inductive hypothesis, there is a path through the Dyck configuration graph that parallels and mimics the path above.
    Thus, there is a configuration $\aconf_1 = (\astate_1, \ass_1)$ such that $\absmap(\conf_1) \wt \aconf_1$.
    Also by this path, there is a sub-path from the configuration $\aconf_1$ to the new configuration $\aconf'$ that makes no changes to the stack.
    Therefore, there is a shortcut edge between these two configurations $\aconf_1$ and $\aconf'$.
    The proof of the first case for shortcut edges works for no-op transitions.
    Thus the stack summaries at these two configurations are the same:
    $\ass_1 = \ass'$.
    
    The current situation is as follows:
    
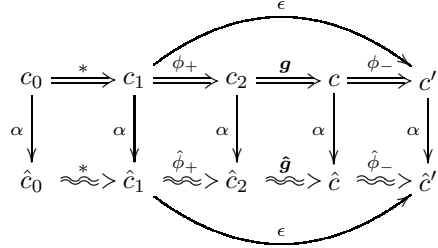
\[
    \xymatrix{
      \conf_0 \ar@{=>}[r]^{\ast} \ar[d]_\absmap
      &
      \conf_1 \ar@{=>}[r]^{\phrame_+} \ar[d]_\absmap \ar@/^2pc/[rrr]^\epsilon
      &
      \conf_2 \ar@{=>}[r]^{\vec{\stackact}} \ar[d]_\absmap
      &
      \conf \ar@{=>}[r]^{\phrame_-} \ar[d]_\absmap
      &
      \conf' \ar[d]_\absmap
      \\
      \aconf_0 \ar@2{~>}[r]^{\ast}
      &
      \aconf_1 \ar@2{~>}[r]^{\aphrame_+} \ar@/_2pc/[rrr]^\epsilon
      &
      \aconf_2 \ar@2{~>}[r]^{\vec{\hat{\stackact}}}
      &
      \aconf \ar@2{~>}[r]^{\aphrame_-}
      &
      \aconf'
    }
    \]
    Since the configuration before the last previously unmatched push, $\conf_1$ is subsumed by its equivalent in the Dyck configuration graph, $\aconf_1$, its stack is subsumed by the stack summary of configuration $\aconf_1$:
    $\ssum(\cont_1) \wts \ass_1$.
    Since this stack and this stack summary are identical to the stack $\cont'$ and the stack summary $\ass'$ respectively, we have:
    $\ssum(\cont') \wts \ass'$.
\qed
\end{itemize}
\end{proof}

\end{document}